\documentclass[twocolumn,showpacs,preprintnumbers,amsmath,amssymb]{revtex4}
\usepackage{mathrsfs}
\usepackage{amsfonts}
\usepackage{}
\usepackage{tikz}

\usepackage{graphicx}% Include figure files
\usepackage{dcolumn}% Align table columns on decimal point
\usepackage{bm}% bold math
\usepackage{amsmath,amsthm}

\usepackage{diagbox}
\begin{document}

\preprint{}

\title
{Local unitary classification of sets of generalized Bell states
in $\mathbb{C}^{d}\otimes\mathbb{C}^{d}$}% Force line breaks with \\

\author{Cai-Hong Wang$^{1}$}%
%\email[]{chwang@cwxu.edu.cn}
\author{Jiang-Tao Yuan$^{1}$}%
\email[]{jtyuan@cwxu.edu.cn}
\author{Mao-Sheng Li$^{2}$}
%\email[]{li.maosheng.math@gmail.com}
\author{Ying-Hui Yang$^{3}$}
%\email[]{yangyinghui4149@163.com}
\author{Shao-Ming Fei$^{4}$}
%\email[]{feishm@cnu.edu.cn}

\affiliation{%
$^{1}$Department of General Education, Wuxi University, Wuxi, 214105, China\\
$^{2}$School  of Mathematics, South China University of Technology, GuangZhou 510640, China\\
$^{3}$School of Mathematics and Information Science, Henan Polytechnic University, Jiaozuo, 454000, China\\
$^{4}$School of Mathematical Sciences, Capital Normal University, Beijing 100048, China\\
}%

\date{\today}% It is always \today, today,

\begin{abstract}
Two sets of quantum entangled states that are equivalent under local unitary transformations may exhibit identical effectiveness and versatility in various quantum information processing tasks. Consequently, classification under local unitary transformations has become a fundamental issue in the theory of quantum entanglement. The primary objective of this work is to establish a complete LU-classification of all sets of generalized Bell states (GBSs) in bipartite quantum systems $\mathbb{C}^{d}\otimes \mathbb{C}^{d}$ with $d\geq 3$. Based on this classification, we determine the minimal cardinality of indistinguishable GBS sets in $\mathbb{C}^{6}\otimes \mathbb{C}^{6}$ under one-way local operations and classical communication (one-way LOCC). We propose first two classification methods based on LU-equivalence for all $l$-GBS sets for $l\geq 2$. We then establish LU-classification for all 2-GBS, 3-GBS, 4-GBS and 5-GBS sets in $\mathbb{C}^{6}\otimes \mathbb{C}^{6}$. Since LU-equivalent sets share identical local distinguishability, it suffices to examine representative GBS sets from equivalent classes.  Notably, we identify a one-way LOCC indistinguishable 4-GBS set among these representatives, thereby resolving the case of $d = 6$ for the problem of determining the minimum cardinality of one-way LOCC indistinguishable GBS sets in [Quant. Info. Proc. 18, 145 (2019)] or [Phys. Rev. A 91, 012329 (2015)].
\end{abstract}

\maketitle

%\tableofcontents

\section{Introduction}
Quantum entanglement plays a key role in quantum information processing and quantum computation. We say that two quantum states in $\mathbb{C}^{d}\otimes\mathbb{C}^{d}$ are local unitary (LU) equivalent if they can be transformed into each other by local unitary operations. Sets of quantum states which are LU-equivalent have the same effectiveness and versatility in many quantum information processing tasks such as quantum secret sharing \cite{raha2015pra,yang2015sci}. And the distinguishability of sets of orthogonal quantum states under local operations and classical communication (LOCC) is exactly the same as the LU equivalence \cite{fan2004prl,nath2005jmp,band2011njp,zhang2014qip,wang2019pra,yang2018pra}. Therefore, an effective way to understand entanglement is to classify entanglement under LU-equivalence. At present, the study on the LU-equivalence of orthogonal quantum state sets is limited to the sets of generalized Bell states (GBSs) in bipartite quantum systems. Since there is a one-to-one correspondence between the maximally entangled states (MESs) and the unitary matrices, specifically, each GBS corresponds to a generalized Pauli matrix (GPM), the LU-equivalence of GBS sets can be expressed in terms of the unitary equivalence (U-equivalence) of GPM sets.

Singal et al. \cite{sing2017pra} first considered the LU-equivalence of 4-GBS sets in $\mathbb{C}^4\otimes \mathbb{C}^4$. All 1820 (=$C_{16}^{4}$) 4-GBS sets in $\mathbb{C}^4\otimes \mathbb{C}^4$ are divided into 122 equivalence classes based on certain specific LU-equivalence. However, this classification is incomplete as different classes may still be LU-equivalent. Subsequently, Tian et al. \cite{tian2016pra} classified all 4-GBS sets in $\mathbb{C}^4\otimes \mathbb{C}^4$ into 10 equivalence classes by using a number of LU-equivalence related to Clifford operators \cite{hos2005pra,fari2014jpa}. This classification is exhaustive, i.e., different classes are not LU-equivalent to each other.
Based on the complete classification, in \cite{yuan2022quantum} the authors completely resolved the local discrimination problem of GBS sets in $\mathbb{C}^4\otimes \mathbb{C}^4$. In \cite{wang2021jmp} Wang et al. discussed the LU-classification of GBS sets in $\mathbb{C}^{5}\otimes\mathbb{C}^{5}$. All 12650(=$C_{25}^{4}$) 4-GBS sets are classified into 8 equivalence classes, and all 53130(=$C_{25}^{5}$) 5-GBS sets are classified into 21 equivalence classes.

Quantum state discrimination is a fundamental problem in quantum information theory. Adding orthogonal states to a LOCC indistinguishable set (or a nonlocal set) gives rise still to a nonlocal set. Therefore, the smaller of the cardinality of a nonlocal set means better construction of a nonlocal set. Hence, it is interesting to ask how small a nonlocal set could be in a given systems. To solve this problem, Zhang et al. \cite{zhang2015pra} introduced the cardinality function $f(d)$, where $f(d)$ is the minimum cardinality of one-way LOCC indistinguishable MES sets in $\mathbb{C}^d\otimes \mathbb{C}^d$, and proved that the minimum cardinality $f(d)$ does not exceed $\frac{d+5}{2}$ ($d$ is odd) or $\frac{d+4}{2}$ ($d$ is even). Wang et al. \cite{wang2016qip} reported that the minimum cardinality $f(d)$ satisfies $f(d)\leq 4$. In this way, important progress has been made towards the minimum cardinality problem of one-way LOCC indistinguishable MES sets. In \cite{yuan2019qip}, the authors considered the cardinality function $f_{GBS}(d)$, where $f_{GBS}(d)$ is the minimum cardinality of one-way LOCC indistinguishable GBS sets in $\mathbb{C}^d\otimes \mathbb{C}^d$. Up to now, only the minimum cardinalities $f_{GBS}(2)=3$, $f_{GBS}(3)=4$, $f_{GBS}(4)=f_{GBS}(5)=4$ and $f_{GBS}(7)=5$ have been obtained \cite{nath2005jmp,gho2004pra,band2011njp,wang2017qip,tian2016pra,yuan2019qip}.
Other minimum cardinalities including $f_{GBS}(6)$ is still unknown.

In this work, we consider the LU-classification of GBS sets in $\mathbb{C}^{d}\otimes \mathbb{C}^{d}$ and determine the exact value of the minimum cardinality $f_{GBS}(6)$. Specifically, two LU-classification methods based on Clifford operators which are applicable to all bipartite systems are proposed. Using these classification methods, all $58905\ (=C_{36}^{4})$ 4-GBS sets (respectively, all 376992 5-GBS sets) in $\mathbb{C}^{6}\otimes \mathbb{C}^{6}$ are classified into 31 equivalent classes (respectively, 112 equivalent classes). By examining the representative elements of the classes of 4-GBS sets, we identify a one-way LOCC indistinguishable 4-GBS set. This set demonstrates that the cardinality of the smallest one-way LOCC indistinguishable GBS sets in $\mathbb{C}^{6}\otimes \mathbb{C}^{6}$ is 4, which implies $f_{GBS}(6)=4$.

The rest of this paper is organized as follows. In Section II, we review the notions of the LU-equivalence of two orthogonal MES set, Clifford operators and related results. In Section III, We present two LU-classification methods based on Clifford operators and provide the LU-classification of all 4-GBS and 5-GBS sets in $\mathbb{C}^{6}\otimes \mathbb{C}^{6}$. In Section IV, we illustrate our results by selecting a one-way LOCC indistinguishable 4-GBS set from the representative elements of the equivalence classes. We draw the conclusions in the last section.

\section{Preliminaries}

\newtheorem{definition}{Definition}
\newtheorem{lemma}{Lemma}
\newtheorem{theorem}{Theorem}
\newtheorem{corollary}{Corollary}
\newtheorem{example}{Example}
\newtheorem{proposition}{Proposition}
\newtheorem{problem}{Problem}
\newtheorem{conjecture}{Conjecture}
\newtheorem{remark}{Remark}

\def\QEDclosed{\mbox{\rule[0pt]{1.3ex}{1.3ex}}}
\def\QED{\QEDclosed}
\def\proof{\indent{\em Proof}.}
\def\endproof{\hspace*{\fill}~\QED\par\endtrivlist\unskip}

\newtheorem{procedure}{Procedure}

Consider a $d$-dimensional Hilbert space $\mathbb{C}^{d}$ with computational basis  $\{|j\rangle\}_{j=0}^{d-1}$. Denote $\mathbb{Z}_{d}=\{0,1,\ldots,d-1\}$. Let $U_{m,n}=X^{m}Z^{n}$, $m, n\in\mathbb{Z}_{d}$, be the generalized Pauli matrices (GPMs) constituting a basis of unitary operators, where $X|j\rangle=|j+1$ mod $d\rangle$, $Z|j\rangle=\omega^{j}|j\rangle$ and $\omega=e^{2\pi i/d}$ are the generalizations of Pauli matrices. The canonical maximally entangled state $|\Phi\rangle$ in $\mathbb{C}^{d}\otimes\mathbb{C}^{d}$ is given by $|\Phi\rangle=(1/\sqrt{d})\sum_{j=0}^{d-1}|jj\rangle$. It is known that $(I\otimes A)|\Phi\rangle=(A^{T}\otimes I)|\Phi\rangle$, where $A^{T}$ denotes the transposition of matrix $A$. Any MES can be written as $|\Psi\rangle=(I\otimes U)|\Phi\rangle$, where $U$ is a unitary. The states
\begin{eqnarray*}
|\Phi_{m,n}\rangle=(I\otimes U_{m,n})|\Phi\rangle
\end{eqnarray*}
are called generalized Bell states (GBSs), where $U=X^{m}Z^{n}$. Note that there is a one-to-one correspondence between the MESs and unitaries. The corresponding unitaries of GBSs are GPMs. For convenience we denote a GBS set $\{ (I\otimes X^{m_i}Z^{n_i})|\Phi\rangle \}$ by $\{  X^{m_i}Z^{n_i} \}$ or $\{  (m_i, n_i) \}$.
For a given $l$-GBS set $\mathcal{S}=\{X^{m_{i}}Z^{n_{i}}|m_{i}, n_{i}\in\mathbb{Z}_{d}, 1\leq i\leq l\}$ with $2\le l\le d$,
the difference set $\Delta \mathcal{S}$ of $\mathcal{S}$ is defined by
\begin{eqnarray*}
&\Delta \mathcal{S}=\{U_{j}U_{k}^{\dag}|U_{j}, U_{k}\in \mathcal{S}, j\neq k\}.
\end{eqnarray*}
Up to a phase, we can identify $\Delta \mathcal{S}$ as the set
\begin{eqnarray*}
\{(m_{j}-m_{k},n_{j}-n_{k})|(m_j, n_j), (m_k, n_k)\in \mathcal{S}, j\neq k\}.
\end{eqnarray*}

According to the definitions of $l$-GBS set $\mathcal{S}$ and its difference set  $\Delta\mathcal{S}$, each element $U_{j}U_{k}^{\dag}$ in $\Delta\mathcal{S}$ are determined by two different GBSs $U_{j}$ and $U_{k}$ in $\mathcal{S}$. Hence, $\Delta\mathcal{S}$ contains at most $l(l-1)$ GPMs. In the set $\Delta \mathcal{S}$, we use the GPM $(m_{j}-m_{k},n_{j}-n_{k})$ to refer to the element $U_{j}U_{k}^{\dag}$ uniformly. It is easy to check that doing so does not change the commutation relationship between the elements in $\Delta\mathcal{S}$.

Let's recall the LU-equivalence of two MES sets.
\begin{definition}[\cite{wu-tian2018pra}]
Let $\{ |\phi_1\rangle , |\phi_2\rangle , \cdots , |\phi_n\rangle \}$ and $\{ |\psi_1\rangle , |\psi_2\rangle , \cdots , |\psi_n\rangle \}$ be two sets of MESs. Their corresponding unitary matrix sets are $\mathcal{M}=\{M_1, M_2, \cdots , M_n \}$ and $\mathcal{N}=\{N_1, N_2, \cdots , N_n \}$. If there exist two unitary operators $U_A$, $U_B$ and a permutation $\sigma$ over $\{ 1, 2, \cdots , n \}$ such that $|\psi_i\rangle \approx (U_A \otimes U_B) |\phi_{\sigma (i)}\rangle$, where $\approx$ denotes ``up to some global phase", then these two MES sets are called LU-equivalent. Meanwhile the corresponding unitary matrix sets are called $U$-equivalent, that is,
\begin{eqnarray*}
L\mathcal{M}R\approx \mathcal{N},
\end{eqnarray*}
where $L=U_B$ and $R=U_A^T$, denoted by $\mathcal{M}\sim \mathcal{N}$. Especially, when $R=L^\dagger$, the two sets $\mathcal{M}$ and $\mathcal{N}$ are called unitary conjugate equivalent (UC-equivalent), denoted by $\mathcal{M}\mathop{\sim}\limits^L \mathcal{N}$.
\end{definition}

In general we can study the LU-equivalence of two MES sets by studying the $U$-equivalence of their unitary matrix sets. For convenience, we call a GBS  set containing the standard MES $|\Phi\rangle$ a standard GBS set, and a GPM set containing the identity matrix $I$ a standard GPM set. Let $\mathcal{M}=\{M_1, M_2, \cdots , M_n \}$ and $\mathcal{N}=\{N_1, N_2, \cdots , N_n \}$ be two GPM sets. If $\mathcal{M}$ is $U$-equivalent to $\mathcal{N}$ and $\mathcal{N}$ is a standard GPM set, then there exist two unitaries $L$, $R$  and a permutation $\sigma$ over $\{ 1, 2, \cdots , n \}$ such that $LM_iR \approx N_{\sigma (i)}$, and $LM_{i_0}R \approx N_{\sigma ({i_0})}=I$ holds for some $i_0$. Thus $M_{i_0}\approx L^{\dagger} R^{\dagger}$ and for any $i$ we have
\begin{align*}
N_{\sigma (i)}&\approx LM_iR =LM_{i_0}M_{i_0}^{\dagger}M_iR\approx R^{\dagger}(M_{i_0}^{\dagger}M_i)R.
\end{align*}
That is to say, the U-equivalence between a GPM set $\mathcal{M}$ and a standard GPM set $\mathcal{N}$ can be achieved by a left multiplication operation (determined by a GPM $M_{i_0}$ in $\mathcal{M}$) and a UC-transformation (determined by a unitary matrix $R$). We have the following useful characterization of U-equivalence of GPM sets (or LU-equivalence of GBS sets).

\begin{lemma}\label{lem2.1}\label{lem2.1}
For a given GPM set $\mathcal{M}=\{M_1, M_2, \cdots , M_n \}$,
the standard GPM set $\mathcal{N}$ that is U-equivalent to the set $\mathcal{M}$ has the form $\mathcal{N}\approx\{R^{\dagger}(M_{i_0}^{\dagger}M_i)R\}$, where $R$ is a unitary matrix.
\end{lemma}

Each GPM set is U-equivalent to a standard GPM set via a left multiplication operation. Let us assume that the GPM sets are all standard GPM sets. In order to find all standard GPM sets that are U-equivalent to a given standard GPM set $\mathcal{M}$, according to Lemma \ref{lem2.1},
it is sufficient to find all standard GPM sets that are UC-equivalent to the standard GPM sets $M_{i_0}^\dagger\mathcal{M}\triangleq \{ M_{i_0}^\dagger M_1, \cdots , M_{i_0}^\dagger M_n \}$ $(i_0=1,\cdots , n)$.

A single-qudit Clifford operator is defined as a unitary operator which maps all GPMs (on $\mathbb{C}^{d}$) to all GPMs under conjugation \cite{hos2005pra,fari2014jpa}. The classical (or symplectic) representation $W$ of a single-qudit Clifford operator is a unique $2\times 2$  symplectic matrix (up to a global phase) over $\mathbb{Z}_d$, that is,
\begin{eqnarray*}
W = \left[
\begin{array}{llll}
a_1 &b_1\\
a_2 &b_2
\end{array}
\right],
\end{eqnarray*}
where the entries are over $\mathbb{Z}_d$ and the determinant of $W$ is 1 (mod $d$). For a GPM $(m,n)$, the Clifford operator $W$ maps $(m,n)$ to $(a_1m+b_1n, a_2m+b_2n)$ (or $X^mZ^n\mathop{\sim}\limits^W X^{a_1m+b_1n}Z^{a_2m+b_2n}$) in a unitary conjugate manner.

Inspired by Hashimoto et al. \cite{hashi2021pra}, we consider the form of LU-operations that transform all GBSs to all GBSs, that is, the unitary matrix $R$ in Lemma \ref{lem2.1} that  maps the standard GPM set $M_{i_0}^\dagger\mathcal{M}$ to GPM set $\mathcal{N}$ in a unitary conjugate manner is a Clifford operator. We refer to the aforementioned LU-operations as the Clifford-operators-based operations. We present two Clifford-operators-based classification methods for GBS (or GPM) sets in $\mathbb{C}^{d}\otimes\mathbb{C}^{d}$ in the next section.

\section{Clifford-operators-based classification for GBS sets in  $\mathbb{C}^{d}\otimes\mathbb{C}^{d}$}\label{sec3}
Since the LU-equivalence between GBS sets is equivalent to the U-equivalence between the corresponding GPM sets, we first provide the necessary and sufficient conditions for two GPMs to be UC-equivalent. Then, we present the classification method for 2-GPM sets and the classification method for 3-GPM sets based on 2-GPM sets. This provides us two classification methods for a general $l$-GPM set ($l>2$). Since a general $l$-GPM set $\mathcal{M}$ is  U-equivalent to the standard $l$-GPM set $(X^{s_0}Z^{t_0})^\dagger\mathcal{M}$ through the left multiplication operation determined by the GPM $(s_0,t_0)$ in $\mathcal{M}$, it can be assumed that $\mathcal{M}$ is a standard $l$-GPM set, that is, $\mathcal{M}=\{(0,0), (s_1,t_1), \cdots,(s_{l-1},t_{l-1})\}$.

\subsection{UC-equivalence of two GPMs}
Recall that the order of a group element $x$ refers to the smallest positive integer $m$ such that $x^{m}=I$, or the infinity if $x^{m}\neq I$ for all $m\neq 0$. We denote $O(x)$ the order of $x$. Since the global phase does not affect the LU-equivalence between quantum states,
we introduce the so-called essential order and the essential power that are more suitable for the study of LU-equivalence.

\begin{definition}
Let $U$ be a GPM defined on $\mathbb{C}^d$.
A positive integer $a$ is said to be the \textit{essential order} of $U$ if it is the smallest positive integer such that $U^a\approx I$. We denote by $O_e(U)$ the \textit{essential order}.
When $U\not\approx I$, we call $d/O_e(U)$ the \textit{essential power} of $U$ and denote it as $P_e(U)$. We set $P_e(U)=0$ when $U\approx I$.
\end{definition}

For each GPM $X^sZ^t$ on $\mathbb{C}^d$, it is easy to verify that its essential order $O_e(X^sZ^t)$ is $\frac{d}{gcd(s,t,d)}$ and,
in turn, its essential power $P_e(X^sZ^t)$ is $gcd(s,t,d)$, where $gcd(s,t,d)$ refers to the largest common divisor of $s, t$ and $d$.

\begin{example}
Let $X$ and $Z$ be Pauli matrices on $\mathbb{C}^2$. Then  $O_e(X)=\frac{2}{gcd(1,0,2)}=2=O(X)$.
Since $(XZ)^2=-I$ and $(XZ)^4=I$, we have $O_e(XZ)=2=\frac{2}{gcd(1,1,2)}$ and $O(XZ)=4$.
\end{example}
\begin{example}
Let $X$ and $Z$ be GPMs on $\mathbb{C}^4$. Then  $O_e(XZ)=\frac{4}{gcd(1,1,4)}=4$.
Since $(XZ)^4=-I$ and $(XZ)^8=I$, we have $O(XZ)=8$.
Similarly, it is easy to know that $O_e(XZ^{3})=4$, $(XZ^{3})^4=-I$ and $O(XZ^{3})=8$.
\end{example}
\begin{example}
Let $X$ and $Z$ be GPMs on $\mathbb{C}^6$. Then $O_e(XZ^{5})=6, (XZ^{5})^6=-I$ and $O(XZ^{5})=12$. Similarly, we have $O_e(X^{3}Z^{3})=2$, $(X^{3}Z^{3})^2=-I$ and $O(X^{3}Z^{3})=4$.
\end{example}

\begin{definition}
Let $\mathcal{M}=\{M_1, M_2, \cdots , M_n \}$ be a GPM set and $x=( P_e(U_1), P_e(U_2), \cdots , P_e(U_n) )^T$ be a $n$-dimensional real vector.
The \textit{essential power vector} of $\mathcal{M}$ is defined as $P_e(\mathcal{M})=x^\uparrow$, where $\uparrow$ stands for that the elements in the vector are arranged in ascending order.
\end{definition}

For example, if $\mathcal{M}=\{X, X^3Z^3, Z^4, X^2 \}$ is a GPM set on $\mathbb{C}^6$,
then $P_e(X)=1$, $P_e(X^3Z^3)=3$, $P_e(Z^4)=2$, $P_e(X^2)=2$ and $P_e(\mathcal{M})=(1, 2, 2, 3)^T$.

\begin{lemma}\label{lem3.1}
Two GPMs are UC-equivalent if and only if they have the same essential power. In particular, if two GPM sets are UC-equivalent, then their essential power vectors are the same.
\end{lemma}

\begin{proof}
For each  GPM $X^sZ^t$ $(\neq I)$ on $\mathbb{C}^d$, let $gcd(s, t)=b$ and $gcd(b, d)=a=P_e(X^sZ^t)$. Then there exist integers $p_1, q_1, p_2, q_2$ such that $p_1 s+q_1 t=b, p_2d+q_2b=a.$ Denote $\overline{C}_{gcd(s,t)}$ and $\overline{C}_{gcd(d,b)}$ as two Clifford operators with symplectic (or classical) representations of
\begin{eqnarray*}
C_{gcd(s,t)}=\left[ \begin{array}{llll} t/b &-s/b\\ p_1 &q_1 \end{array} \right]\ \hbox{and} \
C_{gcd(d,b)}=\left[ \begin{array}{llll} b/a &-d/a\\ p_2 &q_2 \end{array} \right],
\end{eqnarray*}
respectively. We get $X^sZ^t\mathop{\sim}\limits^{\overline{C}_{gcd(s,t)}} Z^b \mathop{\sim}\limits^{\overline{C}_{gcd(d,b)}} Z^a=Z^{P_e(X^sZ^t)}$.
Therefore, two nontrivial GPMs $X^sZ^t$ and $X^{s^{\prime}}Z^{t^{\prime}}$ are UC-equivalent if and only if
$Z^{P_e(X^sZ^t)}$ and $Z^{P_e(X^{s^{\prime}}Z^{t^{\prime}})}$ are UC-equivalent. Since $P_e(X^sZ^t)$ and $P_e(X^{s^{\prime}}Z^{t^{\prime}})$ are factors of $d$ and the essential power of a GPM remains unchanged after a unitary conjugate transformation, it follows that $Z^{P_e(X^sZ^t)}$ and $Z^{P_e(X^{s^{\prime}}Z^{t^{\prime}})}$ are UC-equivalent if and only if $P_e(X^sZ^t)=P_e(X^{s^{\prime}}Z^{t^{\prime}})$, which completes the proof.
\end{proof}

We denote the composition $\overline{C}_{gcd(d,b)}\circ\overline{C}_{gcd(s,t)}$ as $\overline{C}_{gcd(s, t, d)}$, which is a Clifford operator that transforms $X^sZ^t$ to $Z^{gcd(s, t, d)}$. Lemma \ref{lem3.1} shows that two GPM sets with different essential power vectors are not UC-equivalent, and the UC-equivalence between two GPMs is determined by Clifford operators.

\begin{example}
Let $\mathcal{M}=\{X^{12}, Z^3, X^3Z^4, X^5Z^{15} \}$ and $\mathcal{N}=\{X^4Z^6, X^6Z^{12}, X^2, X^3Z^5 \}$ be GPM sets on $\mathbb{C}^{30}$. Then the essential power vectors $P_e(\mathcal{M})=(1, 3, 5, 6)^T$ and $P_e(\mathcal{N})=(1, 2, 2, 6)^T\neq P_e(\mathcal{M})$.
$\mathcal{M}$ and $\mathcal{N}$ are not UC-equivalent to each other.
\end{example}

For any two standard $n$-GPM sets $\mathcal{M}$ and $\mathcal{N}$ on $\mathbb{C}^p$ (where $p$ is an odd prime number),
since the essential power of each nontrivial GPM is 1,
it is obvious that they have the same essential power vector.
However, they are not necessarily UC-equivalent to each other. For example,
the  $5$-GPM sets $\mathcal{M}=\{I, Z, Z^2, Z^3,Z^4\}$ and $\mathcal{N}=\{I, Z, Z^2, X,X^4\}$ on $\mathbb{C}^5$ are  not UC-equivalent \cite[Theorem 4]{wang2021jmp}.
This example illustrates that the UC-equivalence between two GPM sets is more complex than the UC-equivalence between two {\color{blue}GPMs.}

\subsection{Classification of 2-GPM sets}
Let $\mathcal{M}=\{(0,0), (s_1,t_1)\}$ be a standard 2-GPM set. We show how to find all standard 2-GPM sets that are U-equivalent to it through Clifford-operators-based operations (we denote all such GPM sets as $\mathcal{CU}(\mathcal{M}$)).
Since a Clifford-operators-based operation on $\mathcal{M}$ is determined by left multiplication operations and Clifford operators, we provide the following program that can be used to find the Clifford-operators-based equivalence class $\mathcal{CU}(\mathcal{M}$).

\begin{procedure}[Clifford-operators-based equivalence class $\mathcal{CU}(\mathcal{M}$)]\label{proc3.1}
Let $\mathcal{M}= \{(0,0), (s_1,t_1)\}$ be an arbitrary standard 2-GPM set.
\begin{enumerate}
\item[{\rm(1)}] Let $\mathcal{M}_{1}=(s_1,t_1)^{\dag}\mathcal{M}=\{(-s_1,-t_1),(0,0)\}$. Then $\mathcal{M}_{1}$ is a standard GPM set.
\item[{\rm(2)}] Act every Clifford operator on $\mathcal{M}_{i}$ and get all GPM sets that are UC-equivalent to $\mathcal{M}_{i}$ through Clifford operators.
Denote all of the aforementioned standard $2$-GPM sets as $\mathcal{CUC}(\mathcal{M}_{i})$, with $\mathcal{CUC}(\mathcal{M}_{0})$ being $\mathcal{CUC}(\mathcal{M})$.
\item[{\rm(3)}] The set $\mathcal{CU(M)}\triangleq\mathcal{CUC}(\mathcal{M})\cup\mathcal{CUC}(\mathcal{M}_{1})$ is just the set of all standard $2$-GPM sets that
are Clifford-operators-based equivalent to the given set $\mathcal{M}$. Here,
the sets $\mathcal{UC}(\mathcal{M}_{i})$ and $\mathcal{UC}(\mathcal{M}_{j})$ may be the same for different $i$ and $j$.
\end{enumerate}
\end{procedure}

Obviously, the procedure is also applicable to the standard $l (\ge2)$-GPM sets.
Now we show the Clifford-operators-based classification method of all standard 2-GPM sets  $\mathfrak{S}\triangleq\{\{(0,0), (s_1,t_1)\}\big|(s_1,t_1)\in \mathbb{Z}_{d}\times\mathbb{Z}_{d}, (0,0)\neq(s_1,t_1)\}$. The basic idea is to perform Clifford-operators-based classification from small to large in lexicographic order. For example, the 2-GPM set $\{(0,0), (0,1)\}$ is smaller than $\{(0,0), (1,0)\}$.

\begin{procedure}[Clifford-operators-based classification I]\label{proc3.2}
Let $\mathfrak{S}\triangleq\{\{(0,0), (s_1,t_1)\}\big|(s_1,t_1)\in \mathbb{Z}_{d}\times\mathbb{Z}_{d}, (0,0)\neq(s_1,t_1)\}$ be the set composed of all standard 2-GPM sets.
\begin{enumerate}
\item[{\rm(1)}] Obviously, the smallest 2-GPM set is $\mathcal{M}= \{(0,0), (0,1)\}$. Use Procedure \ref{proc3.1} to find the Clifford-operators-based equivalence class $\mathcal{CU}(\mathcal{M})$.
\item[{\rm(2)}] Select the smallest 2-GPM set $\mathcal{M}^{\prime}$ from the remaining 2-GPM sets $\mathfrak{S}\backslash\mathcal{CU}(\mathcal{M})$, and also find the Clifford-operators-based equivalence class $\mathcal{CU}(\mathcal{M}^{\prime})$.
\item[{\rm(3)}] By doing so, we find all the Clifford-operators-based classes of all 2-GPM sets.
\end{enumerate}
\end{procedure}

Obviously, the representative elements of equivalence classes in Program \ref{proc3.2} are all the smallest elements among them. The classification method is also applicable to general cases. We use the above classification method to provide the classification of all 35 standard 2-GPM sets on $\mathbb{C}^6$.

\begin{example}[Classification of 2-GPM sets on $\mathbb{C}^6$]\label{ex3.2.1}
Using Matlab to implement the classification method, we find three equivalence classes with representative elements as follows:
$\mathcal{M}_{1}=\{(0,0),(0,1)\},\mathcal{M}_{2}=\{(0,0),(0,2)\},\mathcal{M}_{3}=\{(0,0),(0,3)\}.$
\begin{enumerate}
\item[{\rm(1)}] The smallest 2-GPM set is $\mathcal{M}_{1}= \{(0,0), (0,1)\}$. Use Procedure \ref{proc3.1} to find the Clifford-operators-based equivalence class $\mathcal{CU}(\mathcal{M}_{1})$. Since every standard 2-GPM set contains an identity matrix, it is actually determined by a nontrivial GPM matrix. According to Lemma \ref{lem3.1}, two GPMs are UC-equivalent if and only if they have the same essential power. The GPM $(0,1)$ in $\mathcal{M}_{1}$ is UC-equivalent to $(0,5)$ in $(0,1)^{\dag}\mathcal{M}_{1}= \{(0,0), (0,5)\}$. So the Clifford-operators-based equivalence class $\mathcal{CU}(\mathcal{M}_{1})$ happens to be the UC-equivalence class of the GPM set $\mathcal{M}_{1}$. See Table \ref{tab3.1} for details.
\item[{\rm(2)}] Select the smallest 2-GPM set $\mathcal{M}_{2}= \{(0,0), (0,2)\}$ from the remaining 2-GPM sets $\mathfrak{S}\backslash\mathcal{CU}(\mathcal{M}_{1})$, and find the Clifford-operators-based equivalence class $\mathcal{CU}(\mathcal{M}_{2})$ through Procedure \ref{proc3.1}.
\item[{\rm(3)}] Select the smallest 2-GPM set $\mathcal{M}_{3}= \{(0,0), (0,3)\}$ from the remaining 2-GPM sets $\mathfrak{S}\backslash(\mathcal{CU}(\mathcal{M}_{1})\cup\mathcal{CU}(\mathcal{M}_{2}))$, and find the last Clifford-operators-based equivalence class $\mathcal{CU}(\mathcal{M}_{3})$ (also an UC-equivalent class).
\end{enumerate}
Therefore, all 35 standard 2-GPM sets on $\mathbb{C}^6$ can be classified into three equivalence classes, represented by elements $\mathcal{M}_{1}= \{(0,0), (0,1)\}$, $\mathcal{M}_{2}= \{(0,0), (0,2)\}$ and $\mathcal{M}_{3}= \{(0,0), (0,3)\}$, respectively (see Table \ref{tab3.1} for details).
\end{example}
%\begin{widetext}
\begin{table}[h]
\caption{\label{tab3.1} Classification of 2-GBS sets in $\mathbb{C}^6\otimes \mathbb{C}^6$}
%\footnotesize
%\begin{ruledtabular}
\begin{tabular}{|c|c|}
%\br
\hline
&\{(0,0),(0,1)\},\{(0,0),(0,5)\},\{(0,0),(1,0)\},\{(0,0),(1,1)\},\\
&\{(0,0),(1,2)\},\{(0,0),(1,3)\},\{(0,0),(1,4)\},\{(0,0),(1,5)\},\\
$\mathcal{CU}(\mathcal{M}_{1})$&\{(0,0),(2,1)\},\{(0,0),(2,3)\},\{(0,0),(2,5)\},\{(0,0),(3,1)\},\\
(24 items)&\{(0,0),(3,2)\},\{(0,0),(3,4)\},\{(0,0),(3,5)\},\{(0,0),(4,1)\},\\
&\{(0,0),(4,3)\},\{(0,0),(4,5)\},\{(0,0),(5,0)\},\{(0,0),(5,1)\},\\
&\{(0,0),(5,2)\},\{(0,0),(5,3)\},\{(0,0),(5,4)\},\{(0,0),(5,5)\}.\\%\mr
\hline
$\mathcal{CU}(\mathcal{M}_{2})$&\{(0,0),(0,2)\},\{(0,0),(0,4)\},\{(0,0),(2,0)\},\{(0,0),(2,2)\},  \\
(8 items)&\{(0,0),(2,4)\},\{(0,0),(4,0)\},\{(0,0),(4,2)\},\{(0,0),(4,4)\}.  \\
\hline
$\mathcal{CU}(\mathcal{M}_{3})$&\{(0,0),(0,3)\},\{(0,0),(3,0)\},\{(0,0),(3,3)\}. \\
\hline
\end{tabular}
%\end{ruledtabular}
\end{table}
%\end{widetext}

It should be noted that since there are only 35 standard 2-GPM sets on $\mathbb{C}^6$, according to Lemma \ref{lem3.1}, the classification can be achieved by directly calculating the essential powers without using Matlab. This can be seen from the essential powers of the GPMs in Table \ref{tab3.1} (The essential powers in the three equivalence classes are 1, 2, and 3).

\subsection{Classification of 3-GPM sets}
The classification of standard 3-GPM sets can be carried out by using the method described in the previous section. If we only need to obtain the representative elements of equivalence classes, we can present a more concise method based on the classification of 2-GPM sets. The representative elements of an equivalence class of 3-GPM sets can be obtained by adding a GPM to the representative elements of an equivalence class of 2-GPM sets. In this way, we can classify all 3-GPM sets more concisely based on the representative elements of the equivalence classes of 2-GPM sets.

\begin{procedure}[Classification method II]\label{proc3.3}
Let GPM sets $\mathcal{M}_{1},\cdots,\mathcal{M}_{k}$ be the smallest representative elements (as in Procedure \ref{proc3.2}) of equivalence classes (of 2-GPM sets) arranged in ascending order. We determine the representative elements of equivalence classes of 3-GPM sets based on these representative elements in order.
\begin{enumerate}
\item[{\rm(1)}] Adding a GPM to $\mathcal{M}_{i}$ $(i=1,\cdots,k)$ yields a 3-GPM set. Denote all such 3-GPM sets as $\mathfrak{S}_{1}$. Obviously, the 3-GPM set $\mathcal{M}_{1}\cup\{(0,2)\}= \{(0,0), (0,1), (0,2)\}$ is the smallest GPM set obtained by adding a GPM to $\mathcal{M}_{1}$. Use Procedure \ref{proc3.1} (applicable to general cases) to find the Clifford-operators-based equivalence class of the set $\mathcal{M}_{1}\cup\{(0,2)\}$ within the range of $\mathfrak{S}_{1}$, i.e., $\mathcal{CU}(\mathcal{M}_{1}\cup\{(0,2)\})\cap \mathfrak{S}_{1}$.
\item[{\rm(2)}] Select the smallest 3-GPM set $\mathcal{M}^{\prime}$ from the remaining 3-GPM sets $\mathfrak{S}_{1}\backslash\mathcal{CU}(\mathcal{M}_{1}\cup\{(0,2)\})$, and also find the Clifford-operators-based equivalence class within $\mathfrak{S}_{1}$, i.e., $\mathcal{CU}(\mathcal{M}^{\prime})\cap \mathfrak{S}_{1}$.
\item[{\rm(3)}] By doing so, we obtain all the representative elements of Clifford-operators-based classes of 3-GPM sets.
\end{enumerate}
\end{procedure}

It should be noted that Procedure \ref{proc3.3} is a classification method for the special 3-GPM sets $\mathfrak{S}_{1}$, rather than all standard 3-GPM sets. So this classification is more concise. Meanwhile, according to the previous analysis, the representative elements of the equivalence classes obtained by doing so is the same as that of Procedure \ref{proc3.2}.

There are a total of 595(=35(35-1)/2) standard 3-GBS sets in $\mathbb{C}^6\otimes \mathbb{C}^6$. The corresponding $\mathfrak{S}_{1}$ contains only 99(=34+33+32) standard 3-GBS sets. Therefore, the classification of $\mathfrak{S}_{1}$ is significantly simpler. See Example \ref{ex3.3.1} for details.

\begin{example}[Classification of 3-GBS sets in $\mathbb{C}^6\otimes \mathbb{C}^6$]\label{ex3.3.1}
Using Matlab to implement the classification method II, we get 9 equivalence classes with representative elements as follows (for the sake of simplicity, we have omitted the common element $(0,0)$):
\begin{eqnarray*}
\mathcal{N}_{1}=\{(0,1),(0,2)\},\mathcal{N}_{2}=\{(0,1),(0,3)\},\\
\mathcal{N}_{3}=\{(0,1),(1,0)\},\mathcal{N}_{4}=\{(0,1),(2,0)\},\\
\mathcal{N}_{5}=\{(0,1),(3,0)\},\mathcal{N}_{6}=\{(0,1),(3,2)\},\\
\mathcal{N}_{7}=\{(0,2),(0,4)\},\mathcal{N}_{8}=\{(0,2),(2,0)\},\\
\mathcal{N}_{9}=\{(0,3),(3,0)\}.
\end{eqnarray*}
\begin{enumerate}
\item[{\rm(1)}] Adding a GPM to $\mathcal{M}_{i}$ $(i=1,\cdots,3)$ in Example \ref{ex3.2.1} yields a 3-GPM set. Denote all 99 such 3-GPM sets as $\mathfrak{S}_{1}$. Obviously, the 3-GPM set $\mathcal{N}_{1}:=\mathcal{M}_{1}\cup\{(0,2)\}= \{(0,0), (0,1), (0,2)\}$ is the smallest GPM set in $\mathfrak{S}_{1}$. Use Procedure \ref{proc3.1} to find the Clifford-operators-based equivalence class of the set $\mathcal{N}_{1}$ within the range of $\mathfrak{S}_{1}$, i.e., $\mathcal{CU}(\mathcal{N}_{1})\cap \mathfrak{S}_{1}$. See Table \ref{tab3.2} for details.
\item[{\rm(2)}] Select the smallest 3-GPM set $\mathcal{N}_{2}:=\{(0,0), (0,1), (0,3)\}$ from the remaining 3-GPM sets $\mathfrak{S}_{1}\backslash\mathcal{CU}(\mathcal{N}_{1})$, and also find the Clifford-operators-based equivalence class $\mathcal{CU}(\mathcal{N}_{2})\cap \mathfrak{S}_{1}$.
\item[{\rm(3)}] By doing so, we find all 9 representative elements of Clifford-operators-based classes of 3-GPM sets, see Table \ref{tab3.2}.
\end{enumerate}
\end{example}
%\begin{widetext}
\begin{table}[h]
\caption{\label{tab3.2} Classification of 99 standard 3-GBS sets in $\mathbb{C}^6\otimes \mathbb{C}^6$ via method II}
%\footnotesize
%\begin{ruledtabular}
\begin{tabular}{|c|c|}
%\br
\hline
$\mathcal{CU}(\mathcal{N}_{1})$&\{(0,1),(0,2)\},\{(0,1),(0,5)\},\{(0,2),(3,1)\},\{(0,2),(3,4)\}.\\
\hline
&\{(0,1),(0,3)\},\{(0,1),(0,4)\},\{(0,2),(0,3)\},\{(0,2),(0,5)\},\\
$\mathcal{CU}(\mathcal{N}_{2})$&\{(0,2),(3,0)\},\{(0,2),(3,2)\},\{(0,2),(3,3)\},\{(0,2),(3,5)\},\\
(22 items)&\{(0,3),(0,4)\},\{(0,3),(0,5)\},\{(0,3),(2,0)\},\{(0,3),(2,1)\},\\
&\{(0,3),(2,2)\},\{(0,3),(2,3)\},\{(0,3),(2,4)\},\{(0,3),(2,5)\},\\
&\{(0,3),(4,0)\},\{(0,3),(4,1)\},\{(0,3),(4,2)\},\{(0,3),(4,3)\},\\
&\{(0,3),(4,4)\},\{(0,3),(4,5)\}.\\%\mr
\hline
$\mathcal{CU}(\mathcal{N}_{3})$&\{(0,1),(1,0)\},\{(0,1),(1,1)\},\{(0,1),(1,2)\},\{(0,1),(1,3)\},  \\
(12 items)&\{(0,1),(1,4)\},\{(0,1),(1,5)\},\{(0,1),(5,0)\},\{(0,1),(5,1)\},  \\
&\{(0,1),(5,2)\},\{(0,1),(5,3)\},\{(0,1),(5,4)\},\{(0,1),(5,5)\}.  \\
\hline
&\{(0,1),(2,0)\},\{(0,1),(2,1)\},\{(0,1),(2,2)\},\{(0,1),(2,3)\},\\
$\mathcal{CU}(\mathcal{N}_{4})$&\{(0,1),(2,4)\},\{(0,1),(2,5)\},\{(0,1),(4,0)\},\{(0,1),(4,1)\},\\
(30 items)&\{(0,1),(4,2)\},\{(0,1),(4,3)\},\{(0,1),(4,4)\},\{(0,1),(4,5)\},\\
&\{(0,2),(1,0)\},\{(0,2),(1,1)\},\{(0,2),(1,2)\},\{(0,2),(1,3)\},\\
&\{(0,2),(1,4)\},\{(0,2),(1,5)\},\{(0,2),(2,1)\},\{(0,2),(2,3)\},\\
&\{(0,2),(2,5)\},\{(0,2),(4,1)\},\{(0,2),(4,3)\},\{(0,2),(4,5)\},\\
&\{(0,2),(5,0)\},\{(0,2),(5,1)\},\{(0,2),(5,2)\},\{(0,2),(5,3)\},\\
&\{(0,2),(5,4)\},\{(0,2),(5,5)\}.\\%\mr
\hline
&\{(0,1),(3,0)\},\{(0,1),(3,1)\},\{(0,1),(3,3)\},\{(0,1),(3,4)\},\\
$\mathcal{CU}(\mathcal{N}_{5})$&\{(0,3),(1,0)\},\{(0,3),(1,1)\},\{(0,3),(1,2)\},\{(0,3),(1,3)\},\\
(20 items)&\{(0,3),(1,4)\},\{(0,3),(1,5)\},\{(0,3),(3,1)\},\{(0,3),(3,2)\},\\
&\{(0,3),(3,4)\},\{(0,3),(3,5)\},\{(0,3),(5,0)\},\{(0,3),(5,1)\},\\
&\{(0,3),(5,2)\},\{(0,3),(5,3)\},\{(0,3),(5,4)\},\{(0,3),(5,5)\}.\\%\mr
\hline
$\mathcal{CU}(\mathcal{N}_{6})$&\{(0,1),(3,2)\},\{(0,1),(3,5)\}\}.\\
\hline
$\mathcal{CU}(\mathcal{N}_{7})$&\{(0,2),(0,4)\}.\\
\hline
$\mathcal{CU}(\mathcal{N}_{8})$&\{(0,2),(2,0)\},\{(0,2),(2,2)\},\{(0,2),(2,4)\},\{(0,2),(4,0)\},\\
(6 items)&\{(0,2),(4,2)\},\{(0,2),(4,4)\}.\\
\hline
$\mathcal{CU}(\mathcal{N}_{9})$&\{(0,3),(3,0)\},\{(0,3),(3,3)\}.\\
\hline
\end{tabular}
%\end{ruledtabular}
\end{table}
%\end{widetext}

There are total of $6545$(=$C^{3}_{35}$) standard 4-GBS sets in $\mathbb{C}^6\otimes \mathbb{C}^6$. The corresponding $\mathfrak{S}_{1}$ only contains no more than 300 standard 4-GBS sets. So the classification of $\mathfrak{S}_{1}$ is even simpler.

Similar to Example \ref{ex3.3.1}, we provide the representative elements of the equivalence class of 4-GBS sets in $\mathbb{C}^6\otimes \mathbb{C}^6$, see the following Example \ref{ex3.3.2} for details.

\begin{example}[Classification of 4-GBS sets in $\mathbb{C}^6\otimes \mathbb{C}^6$]\label{ex3.3.2}
Using Matlab to implement  the classification method II, we find 31 equivalence classes with representative elements (sorted in lexicographic order and presented in complete form), see Table III.
\begin{table}[htbp]\label{tab3.3}
\centering
\caption{Representative elements of 31 equivalent  classes of 4-GBS sets  in $\mathbb{C}^6\otimes \mathbb{C}^6$}
\label{tab5.1}
%\footnotesize
%\begin{ruledtabular}
\begin{tabular}{|c|c|c|c|c|c|}
%\br
\hline
$\mathcal{S}$&Elements of $\mathcal{S}$&$\mathcal{S}$&Elements of $\mathcal{S}$&$\mathcal{S}$&Elements of $\mathcal{S}$\\
\hline
$\mathcal{S}_{1}$&$I,Z,Z^{2},Z^{3}$&$\mathcal{S}_{2}$&$I,Z,Z^{2},Z^{4}$&$\mathcal{S}_{3}$&$I,Z,Z^{2},X$\\
%\mr
\hline
$\mathcal{S}_{4}$&$I,Z,Z^{2},X^{2}$&$\mathcal{S}_{5}$&$I,Z,Z^{2},X^{2}Z$&$\mathcal{S}_{6}$&$I,Z,Z^{2},X^{3}$\\
\hline
$\mathcal{S}_{7}$&$I,Z,Z^{2},X^{3}Z$&$\mathcal{S}_{8}$&$I,Z,Z^{3},Z^{4}$&$\mathcal{S}_{9}$&$I,Z,Z^{3},X$\\
\hline
$\mathcal{S}_{10}$&$I,Z,Z^{3},X^{2}$&$\mathcal{S}_{11}$&$I,Z,Z^{3},X^{2}Z$&$\mathcal{S}_{12}$&$I,Z,Z^{3},X^{3}$\\
\hline
$\mathcal{S}_{13}$&$I,Z,Z^{3},X^{3}Z$&$\mathcal{S}_{14}$&$I,Z,Z^{3},X^{5}$&$\mathcal{S}_{15}$&$I,Z,X,XZ$\\
\hline
$\mathcal{S}_{16}$&$I,Z,X,XZ^{2}$&$\mathcal{S}_{17}$&$I,Z,X,XZ^{3}$&$\mathcal{S}_{18}$&$I,Z,X,X^{2}Z^{2}$\\
\hline
$\mathcal{S}_{19}$&$I,Z,X,X^{3}Z^{5}$&$\mathcal{S}_{20}$&$I,Z,X,X^{4}Z^{4}$&$\mathcal{S}_{21}$&$I,Z,X^{2},X^{2}Z$\\
\hline
$\mathcal{S}_{22}$&$I,Z,X^{2},X^{2}Z^{2}$&$\mathcal{S}_{23}$&$I,Z,X^{2},X^{2}Z^{5}$&$\mathcal{S}_{24}$&$I,Z,X^{2},X^{3}Z$\\
\hline
$\mathcal{S}_{25}$&$I,Z,X^{2},X^{4}$&$\mathcal{S}_{26}$&$I,Z,X^{3},X^{3}Z$&$\mathcal{S}_{27}$&$I,Z,X^{3},X^{3}Z^{3}$\\
\hline
$\mathcal{S}_{28}$&$I,Z,X^{3},X^{3}Z^{5}$&$\mathcal{S}_{29}$&$I,Z^{2},Z^{4},X^{2}$&$\mathcal{S}_{30}$&$I,Z^{2},X^{2},X^{2}Z^{2}$\\
\hline
$\mathcal{S}_{31}$&$I,Z^{3},X^{3},X^{3}Z^{3}$&$$&$$&$$&$$\\
\hline
\end{tabular}
%\end{ruledtabular}
\end{table}
%\end{center}
%\end{widetext}
\end{example}

There are a total of 52360(=$C^{4}_{35}$) standard 5-GBS sets in $\mathbb{C}^6\otimes \mathbb{C}^6$, and the corresponding $\mathfrak{S}_{1}$ contains about 900 standard 5-GBS sets. We provide all the 112 representative elements of the equivalence class of 5-GBS sets in $\mathbb{C}^6\otimes \mathbb{C}^6$ in Appendix A.

\section{Minimum cardinality of one-way LOCC indistinguishable GBS sets in $\mathbb{C}^6\otimes \mathbb{C}^6$}\label{sec4}
Recall that the cardinality function $f_{GBS}(d)$ stands for the minimum cardinality of one-way LOCC indistinguishable GBS sets in $\mathbb{C}^d\otimes \mathbb{C}^d$, whereas $f_{GBS}(6)$ is still unknown. In this section, we prove that $f_{GBS}(6)=4$ based on the classification of 4-GPM sets on $\mathbb{C}^6$ given in Table III.

As GBSs are maximally entangled, their reduced density operators are all proportional to the identity operator, which makes it difficult to distinguish GBSs by using the usual LOCC measurements. Fortunately, we have a sufficient and necessary condition for one-way LOCC discrimination via quantum teleportation \cite{gho2004pra,band2011njp,zhang2015pra}:
{\it An $l$-GBS set $\{|\Phi_{m_{j},n_{j}}\rangle\}_{j=1}^{l}$ in $\mathbb{C}^{d}\otimes\mathbb{C}^{d}$ is one-way LOCC distinguishable
if and only if there exists at least one state $|\alpha\rangle$ for which the states $X^{m_{j}}Z^{n_{j}}|\alpha\rangle$ ($j=1,\cdots,l$) are pairwise orthogonal.}
The key quantum state $|\alpha\rangle$ is generally difficult to identify and is currently typically taken as the eigenvector of a suitable GPM (see \cite{li2022pra,wang2025pra}).
We have the following observation.

{\bf Observation} The 4-GBS set $\mathcal{S}_{30}=\{(0,0),(0,2),(2,0),(2,2)\}$ in Table III is one-way LOCC indistinguishable, and then $f_{GBS}(6)=4$.

The proof of the Observation is given by contradiction. Assume that $\mathcal{S}_{30}$ is one-way LOCC distinguishable. Then there exists one state $|\alpha\rangle$ such that the states $X^{m_{j}}Z^{n_{j}}|\alpha\rangle$ ($X^{m_{j}}Z^{n_{j}}\in \mathcal{S}_{30}$) are pairwise orthogonal. Alice and Bob use the teleportation scheme to teleport the state $|\alpha\rangle$ to Bob via the unknown GBS $|\Phi_{m_i,n_i}\rangle$, which is to be identified. The state in Bob$^{,}$s side after the teleportation is actually $X^{m_{i}}Z^{n_{i}}|\alpha\rangle$. Four possible quantum states $X^{m_{j}}Z^{n_{j}}|\alpha\rangle$ are mutually orthogonal, forming a set of orthogonal relations $\langle\alpha|U_{m,n}|\alpha\rangle=0$ for $(m, n)\in\Delta \mathcal{S}_{30}$, which corresponds to a homogeneous system of equations with the unit vector $|\alpha\rangle$ as a solution. According to the properties of homogeneous systems of equations, if there are enough effective equations, the system will have a unique solution (the trivial solution). This contradicts the existence of a non-zero solution such as $|\alpha\rangle$, leading us to conclude that set $\mathcal{S}_{30}$ is  one-way LOCC indistinguishable. See Appendix B for the complete proof.

\section{Conclusion}\label{sec6}
We have investigated the local unitary classification of all sets of generalized Bell states (GBSs) in bipartite quantum system $\mathbb{C}^{d}\otimes \mathbb{C}^{d}$ ($d\geq 3$), as well as the cardinality problem of the smallest indistinguishable GBS sets in $\mathbb{C}^{6}\otimes \mathbb{C}^{6}$ under one-way local operations and classical communication. The LU-operations considered here are completely determined by the left multiplication and the Clifford operator, namely, the so-called Clifford-operators-based operations. We have presented two Clifford-operators-based classification methods for GBS sets: one is to directly classify all GBS sets, the other one is to classify the larger GBS sets based on the classification of the smaller GBS sets. Using these two classification methods, we have provided the Clifford-operators-based classification of all 2-GBS sets and representative elements of equivalence classes of 3-GBS sets, 4-GBS sets and 5-GBS sets on $\mathbb{C}^{6}$.

It is found that there are total of 9 equivalence classes of 3-GBS sets, 31 equivalence classes of 4-GBS sets, and 112 equivalence classes of 5-GBS sets. The LU-equivalence classification is the foundation for completely solving the local discrimination problem of GBS sets. Therefore, based on the classification of 4-GBS sets, we have shown that the representative element of the 30th equivalence class, i.e., the 4-GBS set $\mathcal{S}_{30}=\{(0,0),(0,2),(2,0),(2,2)\}$, is one-way LOCC indistinguishable. This indicates that the minimum cardinality of one-way LOCC indistinguishable GBS sets in $\mathbb{C}^{6}\otimes \mathbb{C}^{6}$ is 4. A comprehensive investigation of the local discrimination problem of all 4-GBS sets in $\mathbb{C}^{6}\otimes \mathbb{C}^{6}$ requires LU-classification and detectors \cite{wang2025pra}. Our classification methods proposed in this work may also be useful in resolving other related problems such as quantum state discrimination.

\section*{Appendix A: 112 representative elements of 5-GBS sets in $\mathbb{C}^{6}\otimes \mathbb{C}^{6}$ (Table IV)}\label{sec1}
\begin{widetext}
\begin{table*}[hbtp]
\caption{Representative elements of 112 equivalent  classes of 5-GBS sets in $\mathbb{C}^{6}\otimes \mathbb{C}^{6}$ (common element (0,0) is omitted)}
\footnotesize
%\begin{ruledtabular}
\begin{tabular}{|c|}
%\br
\hline
$\{(0,1),(0,2),(0,3),(0,4)\},\{(0,1),(0,2),(0,3),(1,0)\},\{(0,1),(0,2),(0,3),(2,0)\},\{(0,1),(0,2),(0,3),(2,1)\},\{(0,1),(0,2),(0,3),(3,0)\},$\\
$\{(0,1),(0,2),(0,3),(3,1)\},\{(0,1),(0,2),(0,4),(1,0)\},\{(0,1),(0,2),(0,4),(2,0)\},\{(0,1),(0,2),(0,4),(2,1)\},\{(0,1),(0,2),(0,4),(3,0)\},$\\
$\{(0,1),(0,2),(0,4),(3,1)\},\{(0,1),(0,2),(1,0),(1,1)\},\{(0,1),(0,2),(1,0),(1,2)\},\{(0,1),(0,2),(1,0),(1,3)\},\{(0,1),(0,2),(1,0),(2,0)\},$\\
$\{(0,1),(0,2),(1,0),(2,1)\},\{(0,1),(0,2),(1,0),(2,2)\},\{(0,1),(0,2),(1,0),(2,3)\},\{(0,1),(0,2),(1,0),(2,5)\},\{(0,1),(0,2),(1,0),(3,0)\},$\\
$\{(0,1),(0,2),(1,0),(3,1)\},\{(0,1),(0,2),(1,0),(3,2)\},\{(0,1),(0,2),(1,0),(3,3)\},\{(0,1),(0,2),(1,0),(3,4)\},\{(0,1),(0,2),(1,0),(3,5)\},$\\
$\{(0,1),(0,2),(1,0),(4,0)\},\{(0,1),(0,2),(1,0),(4,1)\},\{(0,1),(0,2),(1,0),(4,2)\},\{(0,1),(0,2),(1,0),(4,3)\},\{(0,1),(0,2),(1,0),(4,4)\},$\\
$\{(0,1),(0,2),(1,0),(4,5)\},\{(0,1),(0,2),(1,0),(5,1)\},\{(0,1),(0,2),(1,0),(5,2)\},\{(0,1),(0,2),(1,0),(5,5)\},\{(0,1),(0,2),(2,0),(2,1)\},$\\
$\{(0,1),(0,2),(2,0),(2,2)\},\{(0,1),(0,2),(2,0),(2,3)\},\{(0,1),(0,2),(2,0),(2,5)\},\{(0,1),(0,2),(2,0),(3,0)\},\{(0,1),(0,2),(2,0),(3,1)\},$\\
$\{(0,1),(0,2),(2,0),(3,2)\},\{(0,1),(0,2),(2,0),(4,0)\},\{(0,1),(0,2),(2,0),(4,1)\},\{(0,1),(0,2),(4,2),(2,5)\},\{(0,1),(0,2),(2,0),(4,5)\},$\\
$\{(0,1),(0,2),(2,1),(2,3)\},\{(0,1),(0,2),(2,1),(3,0)\},\{(0,1),(0,2),(2,1),(3,2)\},\{(0,1),(0,2),(2,1),(4,1)\},\{(0,1),(0,2),(3,0),(3,1)\},$\\
$\{(0,1),(0,2),(3,0),(3,2)\},\{(0,1),(0,2),(3,0),(3,3)\},\{(0,1),(0,2),(3,0),(3,5)\},\{(0,1),(0,2),(3,1),(3,4)\},\{(0,1),(0,3),(0,4),(1,0)\},$\\
$\{(0,1),(0,3),(0,4),(2,0)\},\{(0,1),(0,3),(0,4),(3,0)\},\{(0,1),(0,3),(1,0),(1,1)\},\{(0,1),(0,3),(1,0),(1,2)\},\{(0,1),(0,3),(1,0),(1,3)\},$\\
$\{(0,1),(0,3),(1,0),(2,1)\},\{(0,1),(0,3),(1,0),(2,2)\},\{(0,1),(0,3),(1,0),(2,4)\},\{(0,1),(0,3),(1,0),(3,0)\},\{(0,1),(0,3),(1,0),(3,1)\},$\\
$\{(0,1),(0,3),(1,0),(3,3)\},\{(0,1),(0,3),(1,0),(3,4)\},\{(0,1),(0,3),(1,0),(4,2)\},\{(0,1),(0,3),(1,0),(4,3)\},\{(0,1),(0,3),(1,0),(4,5)\},$\\
$\{(0,1),(0,3),(1,0),(5,1)\},\{(0,1),(0,3),(1,0),(5,3)\},\{(0,1),(0,3),(1,0),(5,4)\},\{(0,1),(0,3),(1,0),(5,5)\},\{(0,1),(0,3),(2,0),(2,1)\},$\\
$\{(0,1),(0,3),(2,0),(2,2)\},\{(0,1),(0,3),(2,0),(2,5)\},\{(0,1),(0,3),(2,0),(5,0)\},\{(0,1),(0,3),(2,0),(5,2)\},\{(0,1),(0,3),(2,1),(2,3)\},$\\
$\{(0,1),(0,3),(2,1),(3,0)\},\{(0,1),(0,3),(2,1),(4,1)\},\{(0,1),(0,3),(2,1),(5,0)\},\{(0,1),(0,3),(3,0),(3,1)\},\{(0,1),(0,3),(3,0),(3,3)\},$\\
$\{(0,1),(0,3),(3,0),(5,1)\},\{(0,1),(0,3),(5,0),(5,1)\},\{(0,1),(0,3),(5,0),(5,2)\},\{(0,1),(0,3),(5,0),(5,3)\},\{(0,1),(1,0),(1,1),(2,3)\},$\\
$\{(0,1),(1,0),(1,1),(2,4)\},\{(0,1),(1,0),(1,2),(1,4)\},\{(0,1),(1,0),(1,2),(2,1)\},\{(0,1),(1,0),(1,2),(3,1)\},\{(0,1),(1,0),(1,2),(3,2)\},$\\
$\{(0,1),(1,0),(1,2),(3,3)\},\{(0,1),(1,0),(1,2),(3,5)\},\{(0,1),(1,0),(1,2),(5,3)\},\{(0,1),(1,0),(1,3),(2,2)\},\{(0,1),(1,0),(1,3),(2,3)\},$\\
$\{(0,1),(1,0),(1,3),(4,4)\},\{(0,1),(1,0),(2,2),(2,3)\},\{(0,1),(1,0),(2,2),(4,4)\},\{(0,1),(1,0),(2,2),(4,5)\},\{(0,1),(2,0),(2,1),(4,0)\},$\\
$\{(0,1),(2,0),(2,2),(2,4)\},\{(0,1),(2,0),(2,2),(3,1)\},\{(0,1),(2,0),(2,2),(4,4)\},\{(0,1),(2,0),(3,1),(3,4)\},\{(0,1),(2,0),(3,1),(4,0)\},$\\
$\{(0,2),(0,4),(2,0),(2,2)\},\{(0,2),(0,4),(2,0),(4,0)\}.$\\
\hline
\end{tabular}
%\end{ruledtabular}
\end{table*}
\end{widetext}

\section*{Appendix B: Proof of Observation}\label{sec2}
\begin{proof}
The difference set is given by $\Delta\mathcal{S}_{30}=\{(0,2),(0,4),(2,0),(2,2),(2,4),(4,0),(4,2),(4,4)\}$.
Suppose that $\mathcal{S}_{30}$ is one-way LOCC distinguishable, then there exists a normalized vector $|\alpha\rangle=\sum_{j=0}^{d-1}\alpha_{j}|j\rangle$
such that $\{X^{m_{j}}Z^{n_{j}}|\alpha\rangle\big|(m_{j},n_{j})\in \mathcal{S}_{30}\}$ is an orthonormal set, which implies that $\langle\alpha|U_{m,n}|\alpha\rangle=0$ for $(m, n)\in\Delta \mathcal{S}_{30}$. Let $|\omega_{j}\rangle\triangleq|(1,\omega^{j},\cdots,\omega^{5j})\rangle$ and
$\{|\omega_{j}\rangle\}_{j=0}^{5}$ be a base of $\mathbb{C}^{6}$.
The condition $\Delta \mathcal{S}_{30}\supseteq\{(0,2),(0,4)\}$ implies that
$0=\langle\alpha|Z^{i}|\alpha\rangle=\langle(|\alpha_{j}|^{2})_{j=0}^{d-1}|\omega_{i}\rangle$ for $i=2,4.$ Since $\langle(|\alpha_{j}|^{2})_{j=0}^{d-1}|\omega_{2}\rangle
=\langle(|\alpha_{0}|^{2}+|\alpha_{3}|^{2},|\alpha_{1}|^{2}+|\alpha_{4}|^{2},|\alpha_{2}|^{2}+|\alpha_{5}|^{2})|(1,\omega^{2},\omega^{4})\rangle$
and
$\langle(|\alpha_{j}|^{2})_{j=0}^{d-1}|\omega_{4}\rangle
=\langle(|\alpha_{0}|^{2}+|\alpha_{3}|^{2},|\alpha_{1}|^{2}+|\alpha_{4}|^{2},|\alpha_{2}|^{2}+|\alpha_{5}|^{2})|(1,\omega^{4},\omega^{2})\rangle$,
it can be concluded that the vector
$|(|\alpha_{0}|^{2}+|\alpha_{3}|^{2},|\alpha_{1}|^{2}+|\alpha_{4}|^{2},|\alpha_{2}|^{2}+|\alpha_{5}|^{2})\rangle$
is proportional to $|(1,1,1)\rangle$, and we can set $|\alpha_{i}|^{2}+|\alpha_{i+3}|^{2}=k> 0$, $i=0,1,2$.

The condition $\Delta \mathcal{S}_{30}\supseteq \{(2,0),(2,2),(2,4)\}$ implies that
$0=\langle\alpha|X^{2}Z^{i}|\alpha\rangle=\Sigma_{j=0}^{5}\overline{\alpha}_{j+2}\omega^{ij}\alpha_{j}
=\langle(\overline{\alpha}_{j}\alpha_{j+2})_{j=0}^{5}|\omega_{i}\rangle$ for $i=0,2,4$.
Since $\langle(\overline{\alpha}_{j}\alpha_{j+2})_{j=0}^{5}|\omega_{i}\rangle=
\langle(\overline{\alpha}_{0}\alpha_{2}+\overline{\alpha}_{3}\alpha_{5},\overline{\alpha}_{1}\alpha_{3}+\overline{\alpha}_{4}\alpha_{0},
\overline{\alpha}_{2}\alpha_{4}
+\overline{\alpha}_{5}\alpha_{1})|(1,\omega^{2i},\omega^{4i})\rangle=0$, $i=0,2,4,$
it can be concluded that the vector
$|(\overline{\alpha}_{0}\alpha_{2}+\overline{\alpha}_{3}\alpha_{5},\overline{\alpha}_{1}\alpha_{3}+\overline{\alpha}_{4}\alpha_{0},
\overline{\alpha}_{2}\alpha_{4}+\overline{\alpha}_{5}\alpha_{1})\rangle$ is a zero vector.

Therefore, the three non-zero vectors $|(\alpha_{0},\alpha_{3})\rangle$, $|(\alpha_{2},\alpha_{5})\rangle$ and $|(\alpha_{4},\alpha_{1})\rangle$
are mutually orthogonal in $\mathbb{C}^{2}$, which is impossible and $\mathcal{S}_{30}$ is one-way local indistinguishable.
\end{proof}

\begin{acknowledgments}
This work is supported by NSFC (Grant No. 11971151) and Wuxi University High-level Talent Research Start-up Special Fund.
\end{acknowledgments}

\nocite{*}

\section*{Data availability}
All data needed to evaluate the conclusions in the paper are present in the paper and/or the Supplementary Information.

\section*{Author contributions}
All authors contributed to each part of this work equally.
All authors read and approved the final manuscript.

\section*{Competing interests}
The authors declare no competing interests.

% \bibliography{apssamp}% Produces the bibliography via BibTeX.

\end{document}